\documentclass[leqno]{amsart}

\usepackage{amsmath}
\usepackage{amssymb}
\usepackage{amsthm}
\usepackage{color}
\usepackage[all,arc,curve,color,frame]{xy}
\usepackage{enumerate}

\theoremstyle{definition}
\newtheorem{theorem}{Theorem}[section]

\newtheorem{prop}[theorem]{Proposition}
\newtheorem{definition}[theorem]{Definition}

\newtheorem{example}[theorem]{Example}
\theoremstyle{remark}
\newtheorem{rem}[theorem]{Remark}
\newtheorem*{acknowledgment}{Acknowledgments}

\DeclareMathOperator{\Q}{\mathbb{Q}}
\DeclareMathOperator{\N}{\mathbb{N}}
\DeclareMathOperator{\Z}{\mathbb{Z}}

\DeclareMathOperator{\EE}{\mathbb{E}} 
\DeclareMathOperator{\A}{\mathbb{A}}

\DeclareMathOperator{\Ah}{\widehat{\mathbb{A}}}

\DeclareMathOperator{\sgn}{sgn}

\newcommand{\seq}[1]{{\boldsymbol{#1}}}
\newcommand{\vev}[1]{{\langle#1\rangle}}

\newcommand{\Bvev}[1]{{\Bigl\langle#1\Bigr\rangle}}
\newcommand{\qb}[2]{\genfrac{[}{]}{0pt}{}{#1}{#2}}

\title{Quiver mutation sequences and $q$-binomial identities}

\author{Akishi~Kato}%
\address{Graduate School of Mathematical Sciences,
The University of Tokyo,
3-8-1 Komaba, Meguro-ku, Tokyo 153-8914, Japan.}
\email{akishi@ms.u-tokyo.ac.jp}

\author{Yuma~Mizuno}%
\address{Department of Mathematical and Computing Science,
Tokyo Institute of Technology,
2-12-1 Ookayama, Meguro-ku, Tokyo 152-8550, Japan.}
\email{mizuno.y.aj@m.titech.ac.jp}

\author{Yuji~Terashima}%
\address{Department of Mathematical and Computing Science,
Tokyo Institute of Technology,
2-12-1 Ookayama, Meguro-ku, Tokyo 152-8550, Japan.}
\email{tera@is.titech.ac.jp}

\numberwithin{equation}{section}

\allowdisplaybreaks[4]

\begin{document}

\begin{abstract}
In this paper, first we introduce a quantity called a partition function
for a quiver mutation sequence. The partition function is a generating
function whose weight is a $q$-binomial associated with each mutation. Then,
we show that the partition function can be expressed as a ratio of products
of quantum dilogarithms. This provides a systematic way of
constructing various $q$-binomial multisum identities.
\end{abstract}

\maketitle

\tableofcontents

\section{Introduction}
Quiver mutations appear in various areas of mathematics and physics such
as Donaldson-Thomas theory, low-dimensional topology, representation
theory, and gauge theory. It is very important to capture quantitatively
a common structure hidden in various guises of quiver mutations.

In our previous works \cite{KT1,KT2}, we introduced the \emph{partition
$q$-series} for a quiver mutation loop --- a mutation sequence whose
final quiver is isomorphic to the initial one. The partition $q$-series
is defined as a sum over state variables defined on a graph which describes the
sequential evolution of the quiver, and depend only on the combinatorial
structures of quiver mutation sequences. For certain class of quiver
mutation sequences called reddening sequences, we proved that the graded
version of partition $q$-series essentially coincides with the
combinatorial Donaldson-Thomas invariants introduced by Keller \cite{Keller2011, Keller2012, Keller2013a} motivated
by Nagao \cite{Nagao2011, Nagao2013} and Reineke \cite{Reineke2010, Reineke2011}. This is a combinatorial version of more general
invariants studied in Kontsevich and Soibelman's groundbreaking work \cite{KS2008, KS2010} on
the BPS state counting problem. It goes without saying that cluster
algebras play crucial role in the development such as
the periodicity of $T$- and $Y$-systems and the associated dilogarithm
identities \cite{Fomin2003, IIKKN2013, Keller2013b, Nakanishi2011a, Nakanishi2011b}.

In this paper, we first introduce a quantity called a \emph{partition
function} for a quiver mutation sequence. Just like the partition
$q$-series, the partition functions are defined as a sum over states on
a graph generated by the sequential evolution of the quiver. It has a
$\Z^{n}$-grading which comes from $c$-vectors, as is the case with the
partition $q$-series. The partition function is a generating
function whose coefficients are a product of local weights associated
with mutations.  As compared with the partition $q$-series, however,
there are some important differences. (1) The partition functions are
defined for \emph{mutation sequences} rather than mutation loops. (2)
The partition functions depend explicitly on \emph{initial conditions
on state variables}. (3) The partition functions uses \emph{$q$-binomial
coefficients} as the local weights.

We then show that the partition function is an ``invariant of quiver
mutation sequences'' in the following sense: if two mutation sequences
$\seq{m}$ and $\seq{m}'$ on a quiver $Q$ result in same quivers and
$c$-vectors, then the two partition functions associated to them
coincide, that is, $Z_{\seq{m}} = Z_{\seq{m}'}$.  This is considered to
be a $q$-binomial version of so-called quantum dilogarithm identities.
In fact, we show that the partition
function can be expressed as a ratio of products of quantum
dilogarithms.  Then, we can use the usual quantum dilogarithm
identities to show the invariance of the partition function. 
These results provide a systematic way of constructing various $q$-binomial
multisum identities.

This paper is organized as follows. In Section \ref{sec:museq} we review
some definitions and terminologies about quiver mutations.  In Section
\ref{sec:part func} we define the partition functions of mutation
sequences.  We then give our main results in Section \ref{sec:main}.  We
give a formula that express the partition function as a ratio of
products of quantum dilogarithms, then we prove the identities of the
partition functions.  In Section \ref{sec:examples}, we give some
examples of the partition functions of mutation sequences, which
describe our main results in concrete situations.

\begin{acknowledgment}
 We would like to thank R.~Inoue, R.~Kobayashi, A.~Kuniba, and M.~Yamazaki
 for valuable discussions. This work is partially supported by
 JSPS KAKENHI Grant Number JP16K13752, JP16H03931 and 25400083, and by CREST, JST.
\end{acknowledgment}

\section{Quiver mutation sequences}
\label{sec:museq}

\subsection{Quivers and  mutation sequences}
\label{sec:notation}

We briefly review the notion of quiver mutations introduced by Fomin and Zelevinsky
\cite{Fomin2002}.

A \emph{quiver} $Q$ is an oriented graph.
In this paper we assume that all
quivers are finite oriented graphs without loops
or 2-cycles:
\begin{equation*}
 \text{loop} \quad \vcenter{
  \xymatrix @R=8mm @C=8mm
  @M=2pt{\bullet \ar@(ur,dr)[]}}  \qquad \qquad
  \text{2-cycle}\quad \vcenter{ \xymatrix @R=8mm @C=8mm @M=2pt{\bullet
  \ar@/^/[r] & \bullet \ar@/^/[l]} }.
\end{equation*}
Throughout the paper, we identify the set of vertices of $Q$ with
$\{1,2,\dots,n\}$.  Let $Q_{ij} \in \N=\{0,1,2,\cdots\}$ be a
multiplicity of the arrows,
and consider it as $(i,j)$ entry of an
$n\times n$ matrix. Then the quiver can be identified with a skew-symmetric
integer $n \times n$ matrix $B$ given by
\begin{align}
  B_{ij} = Q_{ij} - Q_{ji}.
\end{align}

Let $k$ be a vertex in a quiver $Q$.  The \emph{mutation} $\mu_k$ is a
transformation of the quiver $Q$ to another quiver $\mu_k(Q)$ as defined
below.
The mutated quiver $\mu_k(Q)$ has the same vertices as $Q$.  Its
arrows are obtained from those of $Q$ via three steps:
\begin{itemize}
 \item[1)] for each length two path $i\to k\to j$, add a new arrow
   $i\to j$;
 \item[2)] reverse the direction of all arrows with source or target $k$;
 \item[3)] remove all $2$-cycles which arose in step 1.
\end{itemize}
Mutations are involutive,
i.e. $\mu_{k}(\mu_{k}(Q))=Q$ for any $1\leq k\leq n$.  The mutations can
be seen as a transformation of skew-symmetric matrices.  Let $B$ be the
skew-symmetric corresponding to the quiver $Q$.  Then the skew-symmetric
matrix $\mu_k(B)$ corresponding to $\mu_k(Q)$ is given by
\begin{equation}
 \label{eq:matrix-mutation} \mu_k(B)_{ij} =
   \begin{cases}
    -B_{ij} &
    \text{if $i=k$ or $j=k$} \\
    B_{ij}+\sgn(B_{ik}) \max(B_{ik}B_{kj},0)
    & \text{otherwise.}
   \end{cases}
\end{equation}

A \emph{mutation sequence} of a quiver $Q$ is a finite sequence of 
mutations starting from $Q$.
We denote it by $\seq{m}=(m_{1},m_{2},\dots,m_{T})$,
a finite sequence of vertices of $Q$.
The mutation sequence $\seq{m}$ induces a (discrete) time evolution of the quivers:
\begin{equation}
\label{eq:Q-seq}
\xymatrix@=20pt{
  Q(0) \ar[r]^-{\mu_{m_{1}}} & Q(1) \ar[r]^-{\mu_{m_{2}}} & \cdots \ar[r]
  & Q(t-1) \ar[r]^-{\mu_{m_{t}}} & Q(t) \ar[r] &\cdots \ar[r]^-{\mu_{m_{T}}} & Q(T)
},
\end{equation}
where $Q(t) := \mu_{m_{t}}(Q(t-1))$.
The quiver $Q(0)$ is called the \emph{initial} quiver
and $Q(T)$ the \emph{final} quiver. We will use the notation
$\mu_{\seq{m}}(Q)=\mu_{m_T}(\cdots \mu_{m_2}(\mu_{m_1}(Q))\cdots)$.

\subsection{Ice quivers and $c$-vectors}
\label{sec:c-vectors}

We will follow the terminology in \cite{Bruestle2013}. An \emph{ice
quiver} is a pair $(\widetilde{Q},F)$ where $\widetilde{Q}$ is a quiver
and $F$ is a (possibly empty) subset of
vertices of $\widetilde{Q}$ such that there are no arrows
between them.
Vertices in $F$ are called \emph{frozen vertices}.
Two ice quivers $(\widetilde{Q},F)$ and
$(\widetilde{Q}',F')$ are \emph{frozen isomorphic} if $F=F'$ and
there is an isomorphism of quivers $\phi : \widetilde{Q}\to
\widetilde{Q}'$ such that $\phi|F$ is an identity map.

For any quiver $Q$, we can canonically construct an ice quiver $Q^{\wedge}$
called a $\emph{framed quiver}$. This is obtained from $Q$ by adding, for each vertex $i$, a new
frozen vertex $i'$ and a new arrow $i\to i'$.

Let $\seq{m}=(m_{1},m_{2},\dots,m_{T})$ be a mutation sequence of
$Q$. By putting
\begin{equation}
 \widetilde{Q}(0)=Q^{\wedge}, \qquad
  \widetilde{Q}(t)=\mu_{m_{t}}(\widetilde{Q}(t-1))\qquad (t=1,2,\dots,T)
\end{equation}
we can construct a sequence of ice quivers
\begin{equation}
\label{eq:Q-tilde-seq}
\xymatrix@=20pt{
  \widetilde{Q}(0) \ar[r]^-{\mu_{m_{1}}} & \widetilde{Q}(1) \ar[r]^-{\mu_{m_{2}}} & \cdots \ar[r]
  & \widetilde{Q}(t-1) \ar[r]^-{\mu_{m_{t}}} & \widetilde{Q}(t) \ar[r] &\cdots 
  \ar[r]^-{\mu_{m_{T}}} & \widetilde{Q}(T)
}.
\end{equation}
Note that we never mutate at frozen vertices $F=\{1',\dots,n'\}$.
Let $\widetilde{B}(t)$ denote the skew-symmetric
matrix corresponding to $\widetilde{Q}(t)$.

The $c$-vectors, which were introduced by Fomin and Zelevinsky \cite{Fomin2007},
are defined by counting arrows to/from frozen vertices.
\begin{definition}
  \label{def:c-vector} A \emph{$c$-vector} of a vertex $v$ in $Q(t)$ is a
  vector in $\Z^{n}$ given by
  \begin{equation}
    c_{v}(t):=\bigl(\widetilde{B}(t)_{vi'}\bigr)_{i=1}^{n}.
  \end{equation}
\end{definition}
If the vertices of $\widetilde{Q}(t)$ are ordered as $(1,\dots,n,
1',\dots,n')$, the skew-symmetric matrix $\widetilde{B}(t)$ has the
block form
\begin{equation}
 \label{eq:B-tilde-block}
 \def\h{\rule[-2.2ex]{0pt}{5.8ex}}
  \widetilde{B}(t) =
  \begin{array}{|c|c|}
   \hline
    \h B(t) & ~C(t)~ \\
   \hline
    \h {-}C(t)^{\top} & 0 \\ \hline
  \end{array}\,,\qquad  C(t)=\begin{array}{|@{\qquad}c@{\qquad}|}
  \hline c_{1}(t) \\
  \hline c_{2}(t) \\
  \hline \cdots  \\
  \hline c_{n}(t) \\
  \hline
      \end{array}\;,
\end{equation}
where $X^{\top}$ denotes the transpose of $X$.  The $n\times n$ block
$C(t)$ is called \emph{$c$-matrix}, which consists of row of
$c$-vectors.  By construction, $c_{i}(0)=e_{i}$, where $e_{i}$ is the
$i$-th standard unit vector in $\Z^{n}$.

\subsection{Green and red mutations}
\label{sec:green sequences}

The vertex $v$ of $Q(t)$ is called \emph{green}
(resp. \emph{red}) if $c_{v}(t)\in \N^{n}$ (resp. $-c_{v}(t)\in
\N^{n}$). Note that every vertex of the initial quiver $Q(0)$ is green
by construction.  The following crucial property is satisfied by
$c$-vectors:
\begin{theorem}[Sign coherence \cite{Derksen2010,Plamondon2011,Nagao2013}]
Any vertex $v$ in $Q(t)$ is either green or red.
\end{theorem}

The mutation $ \mu_{m_{t}}:Q(t-1)\to
Q(t)$ is \emph{green} (resp. \emph{red}) if the mutating vertex $m_{t}$
is green (resp. red) on $Q(t-1)$, i.e. on the quiver before
mutation.
The \emph{sign} $\varepsilon_{t}$ of the mutation
$\mu_{m_{t}}$ is defined by
\begin{equation}
 \label{eq:eps-def}
 \varepsilon_{t} =
  \begin{cases}
   +1 & \text{if $\mu_{m_{t}}$ is green}, \\
   -1 & \text{if $\mu_{m_{t}}$ is red}. \\
  \end{cases}
\end{equation}

\begin{figure}[tb]
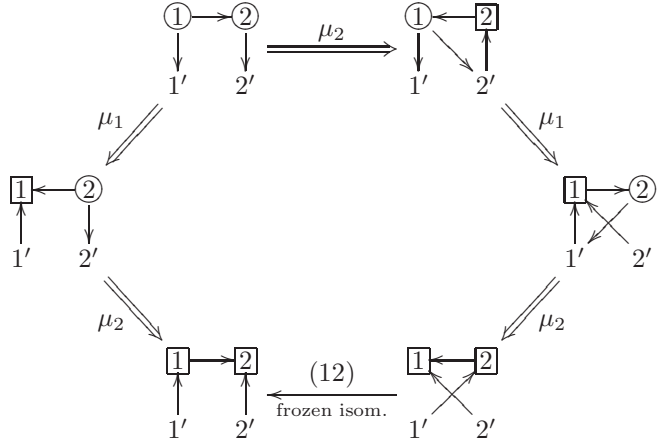

\begin{equation*}
 \xygraph{!{<0cm,0cm>;<16mm,0cm>:<0cm,23mm>::}
 !{(1,2)}*+{
  \vcenter{\xybox{\xygraph{!{<0cm,0cm>;<9mm,0cm>:<0cm,-9mm>::}
 !{(0,0)}*+<6pt>[Fo]{1}="v1"
 !{(1,0)}*+<6pt>[Fo]{2}="v2"
 !{(0,1)}*+<3pt>{1'}="w1"
 !{(1,1)}*+<3pt>{2'}="w2"
 "v1":"v2"
 "v1":"w1"
 "v2":"w2"
 }}}}="Q1"
 !{(-0.3,1)}*+{
  \vcenter{\xybox{\xygraph{!{<0cm,0cm>;<9mm,0cm>:<0cm,-9mm>::}
 !{(0,0)}*+<3pt>[F]{1}="v1"
 !{(1,0)}*+<6.0pt>[Fo]{2}="v2"
 !{(0,1)}*+<3pt>{1'}="w1"
 !{(1,1)}*+<3pt>{2'}="w2"
 "v2":"v1"
 "w1":"v1"
 "v2":"w2"
 }}}}="Q2"
 !{(1,0)}*+{
  \vcenter{\xybox{\xygraph{!{<0cm,0cm>;<9mm,0cm>:<0cm,-9mm>::}
 !{(0,0)}*+<3pt>[F]{1}="v1"
 !{(1,0)}*+<3pt>[F]{2}="v2"
 !{(0,1)}*+<3pt>{1'}="w1"
 !{(1,1)}*+<3pt>{2'}="w2"
 "v1":"v2"
 "w1":"v1"
 "w2":"v2"
 }}}}="Q3"
 !{(3,2)}*+{
  \vcenter{\xybox{\xygraph{!{<0cm,0cm>;<9mm,0cm>:<0cm,-9mm>::}
 !{(0,0)}*+<6.0pt>[Fo]{1}="v1"
 !{(1,0)}*+<3pt>[F]{2}="v2"
 !{(0,1)}*+<3pt>{1'}="w1"
 !{(1,1)}*+<3pt>{2'}="w2"
 "v2":"v1"
 "v1":"w1"
 "v1":"w2"
 "w2":"v2"
 }}}}="Q4"
 !{(4.3,1)}*+{
  \vcenter{\xybox{\xygraph{!{<0cm,0cm>;<9mm,0cm>:<0cm,-9mm>::}
 !{(0,0)}*+<3pt>[F]{1}="v1"
 !{(1,0)}*+<6.0pt>[Fo]{2}="v2"
 !{(0,1)}*+<3pt>{1'}="w1"
 !{(1,1)}*+<3pt>{2'}="w2"
 "v1":"v2"
 "v2":"w1"
 "w2":"v1"
 "w1":"v1"
 }}}}="Q5"
 !{(3,0)}*+{
  \vcenter{\xybox{\xygraph{!{<0cm,0cm>;<9mm,0cm>:<0cm,-9mm>::}
 !{(0,0)}*+<3pt>[F]{1}="v1"
 !{(1,0)}*+<3pt>[F]{2}="v2"
 !{(0,1)}*+<3pt>{1'}="w1"
 !{(1,1)}*+<3pt>{2'}="w2"
 "v2":"v1"
 "w2":"v1"
 "w1":"v2"
 }}}}="Q6"
 "Q1":@{=>}_{\textstyle \mu_{1}}"Q2"
 "Q2":@{=>}_{\textstyle \mu_{2}}"Q3"
 "Q1":@{=>}^{\textstyle \mu_{2}}"Q4"
 "Q4":@{=>}^{\textstyle \mu_{1}}"Q5"
 "Q5":@{=>}^{\textstyle \mu_{2}}"Q6"
 "Q3":@{<-}^{\textstyle (12)}_{\text{frozen isom.}}"Q6"
}
\end{equation*}
\caption{Pentagon and the $A_{2}$ quiver. The green and red vertices are
 marked with circles and boxes, respectively. Both
 $\seq{m}=(1,2)$ and $\seq{m}'=(2,1,2)$ are maximal green sequences.}
 \label{fig:pentagon}
\end{figure}

A mutation sequence $\seq{m}=(m_{1},m_{2},\dots,m_{T})$
is called a \emph{green sequence} if $m_{t}$
is green for all $t$.
A mutation sequence $\seq{m}$ is called a \emph{reddening sequence} if all
vertices of the final quiver $Q(T)$ are red.
A mutation sequence $\seq{m}$ is called a \emph{maximal green sequence} if it is green and reddening.
All maximal green sequences are reddening by definition, but there is a reddening sequence that is not maximal green (see Example \ref{example:B_2}).
In Figure \ref{fig:pentagon}, the two maximal green
sequences $(1,2)$ and $(2,1,2)$  are shown for the $A_{2}$ quiver.

\subsection{Noncommutative algebra $\Ah_{Q}$}
We introduce a noncommutative associative algebra which we use in the definition of the partition function.

Let $Q$ be a quiver with vertices $\{1,2,\dots,n\}$. We define a skew
symmetric bilinear form $\vev{~,~}:\Z^{n}\times \Z^{n}\to \Z$ by
\begin{equation}
 \label{def:vev}
  \vev{e_{i},e_{j}}=B_{ij}=-B_{ji}=Q_{ij}-Q_{ji},
\end{equation}
where $e_{1},\dots,e_{n}$ are the standard basis vectors in $\Z^{n}$.

Let $\A_{Q}$ be a noncommutative associative algebra
over $\Q(q^{1/2})$ presented as
\begin{equation}
 \label{eq:quantum-affine-space}
  \A_{Q}=\Q(q^{\frac{1}{2}})
  \langle\,
  y^{\alpha},~\alpha\in \N^{n} ~|~
  y^{\alpha}y^{\beta}
  = q^{\frac{1}{2}\vev{\alpha,\beta}}y^{\alpha+\beta}\,
  \rangle.
\end{equation}
Its completion with respect to the $\N^{n}$-grading is denoted by
$\widehat{\A}_{Q}$, which can be regarded as the ring of noncommutative power series
in $y_{i}:=y^{e_{i}}$ $(i=1,\dots,n)$.

We will often use the following relations
($\alpha=(\alpha_1,\dots,\alpha_n ) \in \Z^n $):
\begin{equation}
 \label{eq:y-monom}
\begin{split}
 & y_{1}^{\alpha_{1}}
  y_{2}^{\alpha_{2}}
  \dots
  y_{n}^{\alpha_{n}}
  =q^{\frac{1}{2}\sum_{i<j} B_{ij}\alpha_{i}\alpha_{j}}
  y^{\alpha},
 \\
 & y^{\alpha}=q^{-\frac{1}{2}\sum_{i<j} B_{ij}\alpha_{i}\alpha_{j}}
  y_{1}^{\alpha_{1}}
  y_{2}^{\alpha_{2}}
  \dots
  y_{n}^{\alpha_{n}}.
\end{split}
\end{equation}
The $q$-factors in the definition of $y^{\alpha}$ are chosen to guarantee
invariance under the reindexing of the $y_{i}$'s.

\section{Partition functions}
\label{sec:part func}
In this section, we introduce partition functions for mutation sequences.

Let $Q$ be a quiver with vertices $\{1,2,\dots,n\}$.  We consider a
mutation sequence $\seq{m}=(m_{1},m_{2},\dots,m_{T})$ of $Q$:
\begin{equation}
\label{eq:Q-seq-2}
\xymatrix@=20pt{
  Q(0) \ar[r]^-{\mu_{m_{1}}} & Q(1) \ar[r]^-{\mu_{m_{2}}} & \cdots \ar[r]
  & Q(t-1) \ar[r]^-{\mu_{m_{t}}} & Q(t) \ar[r] &\cdots \ar[r]^-{\mu_{m_{T}}} & Q(T)
}.
\end{equation}

We first give a family of \emph{$s$-variables} $\{s_{i}\}$,
\emph{$k$-variables} $\{k_{t}\}$, and
\emph{$k^{\vee}$-variables} $\{k^{\vee}_{t}\}$ by the following rule:
\begin{itemize}
 \item[(i)] An ``initial'' $s$-variable $s_{v}$ is attached to each
      vertex $v$ of the initial quiver $Q=Q(0)$.

 \item[(ii)] When we mutate a quiver at vertex $v$, we add a ``new''
      $s$-variable associated with $v$.

 \item[(iii)] We associate $k_{t}$ and $k^{\vee}_{t}$ with each mutation
      at $m_{t}$.

\end{itemize}

The $s$-, $k$-, and $k^{\vee}$-variables are not regarded as
independent. We require a linear relation for each mutation.
If the quiver $Q(t-1)$ equipped with $s$-variables $\{s_{i}\}$
is mutated at vertex $v=m_{t}$ to give $Q(t)$, then $k$- and
$s$-variables are required to satisfy
\begin{equation}
 \label{eq:k-s-rel}
 k_{t}=\begin{cases}
  \displaystyle
  s_{v}+s_{v}'-\sum_{a\to v} s_{a} & \text{if $\mu_{v}$ is
  green} ~(\varepsilon_{t}=1)\\
  \displaystyle \rule{0pt}{15pt}
  \sum_{v\to b} s_{b} - (s_{v}+s_{v}') & \text{if $\mu_{v}$ is
  red} ~(\varepsilon_{t}=-1)\\
       \end{cases}
\end{equation}
Here, $s'_{v}$ is the ``new'' $s$-variable attached to mutated vertex
$v$, and the sum is over all the arrows of $Q(t-1)$.
Similarly, $k^{\vee}$- and $s$-variables satisfy
\begin{equation}
 \label{eq:kv-s-rel}
 k^{\vee}_{t}
  =\begin{cases}
    \displaystyle
    \sum_{v\to b} s_{b} - (s_{v}+s_{v}') & \text{if $\mu_{v}$ is
  green} ~(\varepsilon_{t}=1)\\
  \displaystyle \rule{0pt}{15pt}
    s_{v}+s_{v}' -\sum_{a\to v} s_{a}
 & \text{if $\mu_{v}$ is
  red} ~(\varepsilon_{t}=-1)\\
   \end{cases}
\end{equation}
Therefore,
\begin{equation}
 \label{eq:kv-k}
  k_{t} + k^{\vee}_{t}= \sum_{v\to b} s_{b} - \sum_{a\to v} s_{a}
\end{equation}
holds at each mutation.

It is sometimes useful to view the relation \eqref{eq:k-s-rel} as
a discrete time evolution of $s$-variables with the control parameters
$\{k_{t}\}$.  Let $s_{i}(t)$ be the $s$-variable
associated with the vertex $i$ in $Q(t)$. Then \eqref{eq:k-s-rel} can be
written as
\begin{equation}
 \label{eq:s-evolution}
 s_{i}(t)=
  \begin{cases}
   \rule{0pt}{14pt} 
   s_{i}(t{-}1) & \text{if $i\neq v$},
   \\
   \displaystyle \rule{0pt}{19pt}
   k_{t}-s_{v}(t{-}1)+\sum_{a}Q(t)_{a,v} s_{a}(t{-}1)
   & \text{if $i=v$ and $\mu_{v}$ is green},
   \\
   \displaystyle \rule{0pt}{13pt}
   -k_{t}-s_{v}(t{-}1)+\sum_{b}Q(t)_{v,b} s_{b}(t{-}1)
   & \text{if $i=v$ and $\mu_{v}$ is red}.
  \end{cases}
\end{equation}
With this notation, \eqref{eq:kv-k} reads as
\begin{equation}
 \label{eq:k-kv-rel}
   k_{t} + k_{t}^{\vee}=\sum_{i} B(t{-}1)_{v,i} \,s_{i}(t{-}1)
   = - \sum_{i} B(t{-}1)_{i,v}\, s_{i}(t{-}1).
\end{equation}
Consequently, we have
\begin{prop}
  \label{prop:kv-as-k-and-r}
Each $k^{\vee}$-variable is written as
a $\Z$-linear combination of $k$-variables
and initial $s$-variables.
\end{prop}
Hereafter, we write $r_i$ for the initial $s$-variables $s_i(0)$ at the vertex $i$ of $Q$ for $i \in \{1,2, \dots ,n\}$.

We now introduce a mutation weight.
\begin{definition}
 Consider a mutation $\mu_{v}$ at the
 vertex $v$. Then
 the \emph{mutation weight} of $\mu_{k}$ is given by
  \begin{equation}
    \label{eq:mutation weight}
    W^{\varepsilon}(k,k^{\vee})= q^{-\frac{1}{2}\varepsilon k k^{\vee}}
    \qb{k +k^{\vee}}{k}_{q^{\varepsilon}}.
  \end{equation}
 Here, $\varepsilon\in \{1,-1\}$ is the sign of
 $\mu_{v}$ and $k\in \N$, $k^{\vee}\in \Z$ are $k$-, $k^{\vee}$-variables 
 associated with the mutation $\mu_{v}$, respectively. The $q$-binomial coefficient is
 defined by
  \begin{equation}
    \label{eq:q-binomial coefficient}
    \qb{m}{k}_{q} = \frac{(q^{m-k+1};q)_{k}}{(q;q)_{k}}
  \end{equation}
  with $(x;q)_{k}  := \prod_{i=0}^{k-1}(1-xq^{i})$.
\end{definition}
\begin{rem}
  $W^{\varepsilon}(k,k^{\vee})$ is a Laurent polynomial of $q^{1/2}$.
  Note that we allow $k^{\vee}$ to be negative integer since our definition of $q$-binomial coefficient \eqref{eq:q-binomial coefficient}
  is valid even if $m<k$.
\end{rem}

Let $\alpha_{t}\in \N^{n}\setminus\{0\}$ denote the \emph{sign-corrected
$c$-vector}
\begin{equation}
  \label{eq:alpha-def}
 \alpha_{t}=\varepsilon_{t}c_{m_{t}}(t{-}1),
\end{equation}
namely, the product of the sign and the $c$-vector of the vertex on which
the $t$-th mutation takes place.

\begin{definition}
  Consider a quiver mutation sequence $\seq{m}=(m_{1},\cdots,m_{T})$
  with an initial quiver $Q$.  The \emph{partition function} of
  $\seq{m}$ with initial $s$-variables $r=(r_1, \dots ,r_n)\in \Z^{n}$
  is given by
  \begin{equation}
    \label{eq:partition function}
    Z_{\seq{m}}(r)=\sum_{k\geq 0} \biggl(
    \prod_{t=1}^{T}
    W^{\varepsilon_{t}}(k_{t},k_{t}^{\vee})
    \biggr)\;
    y^{\sum_{t=1}^{T}k_{t}\alpha_{t}}
    \quad \in \Ah_{Q}
  \end{equation}
 Here the sum is taken over $k=(k_1,\dots,k_T) \in \N^{T}$, and
 all the mutation weights
  $\{W^{\varepsilon_t}(k_t,k_t^{\vee})\}_{1\leq t \leq T}$ are
  considered as functions of $k \in \N^{T}$ and initial
  $s$-variables $r \in \Z^{n}$
  via Proposition \ref{prop:kv-as-k-and-r}.  Consequently, $Z_{\seq{m}}$
  defines a map $Z_{\seq{m}} : \Z^{n} \to \Ah_{Q}$, $r\mapsto
  Z_{\seq{m}}(r)$.
\end{definition}

\begin{example}
  \label{example:A2-part-poly}
  Let us consider the $A_2$ quiver
  \begin{equation}
    \label{eq:A2-quiver}
    Q=\vcenter{
    \xymatrix @R=6mm @C=6mm @M=4pt{
    1 \ar[r] & 2}
    }
  \end{equation}
  and a mutation sequence $\seq{m}=(1,2)$.
  See the left half of Figure \ref{fig:pentagon}.
  We with to compute the partition function $Z_{\seq{m}}(r_1,r_2)$.
  The sign-corrected $c$-vector $\alpha_t$ and the sign $\varepsilon_t$ are computed to be
  \begin{equation}
    \begin{split}
      & \alpha_1 = (1,0), \quad \varepsilon_1 = +, \\
      & \alpha_2 = (0,1), \quad \varepsilon_2 = +.
    \end{split}
  \end{equation}
  The $s$-variables change as follows (cf. \eqref{eq:s-evolution}):
  \begin{align*}
    \begin{array}{c|c|c}
          & s_1(t)        & s_2(t)                \\ \hline
      t=0 & r_1           & r_2                   \\
      t=1 & k_1-r_1       & r_2                   \\
      t=2 & k_1-r_1       & k_2-r_2               \\
    \end{array}
  \end{align*}
  The $k^{\vee}$-variables are given by
  \begin{equation}
    \begin{split}
      k_1^{\vee} &= -s_1(0)-s_1(1)+s_2(0) = -k_1 +r_2 \\
      k_2^{\vee} &= -s_2(1)-s_2(2)+s_1(1) = k_1 -k_2 -r_1 .
    \end{split}
  \end{equation}
  Putting these into the definition of mutation weight \eqref{eq:mutation weight} and
  summing over $k$-variables, we obtain
  \begin{equation}
    \begin{split}
      Z_{\seq{m}}(r_1,r_2) &= \sum_{k_1,k_2\geq 0} W(k_1,-k_1+r_2)W(k_2,k_1-k_2-r_1) y^{(k_1,k_2)} \\
      &= \sum_{k_1,k_2\geq 0} q^{\frac{1}{2}(k_1^2+k_2^2 - k_1 k_2 - k_1 r_2 + k_2 r_1)}
      \qb{r_2}{k_1}_{q} \qb{k_1-r_1}{k_2}_{q} y^{(k_1,k_2)}  .
    \end{split}
  \end{equation}
  For example,
 \begin{equation*}
  \begin{aligned}
      Z_{\seq{m}}(-2,1)=& 1 + y^{(1,0)} + (q^{-\frac{1}{2}} +
   q^{\frac{1}{2}}) y^{(0,1)} 
   \\
   &+(q^{-1} + 1 +q) y^{(1,1)} + y^{(0,2)}
   + (q^{-1} + 1 +q) y^{(1,2)} + y^{(1,3)} , 
   \\ Z_{\seq{m}}(2,2) =& 1 +
   (q^{-\frac{1}{2}} + q^{\frac{1}{2}}) y^{(1,0)} + (-q^{-\frac{1}{2}} -
   q^{\frac{1}{2}}) y^{(0,1)} \\ &+y^{(2,0)} + (-q^{-\frac{1}{2}} -
   q^{\frac{1}{2}}) y^{(1,1)} +(q^{-1} + 1 +q) y^{(0,2)} + \cdots.
\end{aligned} 
 \end{equation*}
  Note that the sum in $Z_{\seq{m}}(-2,1)$ is actually finite, whereas the sum in $Z_{\seq{m}}(2,2)$ is infinite.
\end{example}

\section{Main Results}
\label{sec:main}
\subsection{Statement of the main results}
Our first main result relates the partition functions to the products of
quantum dilogarithms
\begin{equation}
 \EE(x) =\EE(x;q) =\prod_{n=0}^{\infty} \frac{1}{1+q^{n+\frac{1}{2}}x}.
\end{equation}
\begin{theorem}
 \label{thm:main-1}
 Let $[x^n]F(x)$ denote the coefficient of $x^{n}$ in the series $F(x) =
 \sum a_{n} x^{n}$.
Then, for any $\beta\in \Z^{n}$, we have
  \begin{equation*}
    \boxed{
      \begin{split}
        & q^{\frac{1}{2}\vev{\beta,r}}  ([y^{\beta}]Z_{\seq{m}} (r))    \\
        &= [y^{\beta}]
        \left(
        (
        \EE(y^{\alpha_{1}};q^{\varepsilon_{1}})
        \cdots
        \EE(y^{\alpha_{T}};q^{\varepsilon_{T}})
        )^{-1}
        (
        \EE(q^{\vev{\alpha_1,r}} y^{\alpha_{1}};q^{\varepsilon_{1}})
        \cdots
        \EE(q^{\vev{\alpha_{T},r}}y^{\alpha_{T}};q^{\varepsilon_{T}})
        )
        \right)
      \end{split}
    }
  \end{equation*}
  where $\vev{~,~}$ is a skew symmetric bilinear form defined by
  \eqref{def:vev}.
\end{theorem}

Thanks to Theorem \ref{thm:main-1}, any identity between products of
quantum dilogarithm series leads to that for the partition functions.
In particular, applying this result to the mutation sequences defining
combinatorial Donaldson-Thomas invariants, we obtain our second main result:
\begin{theorem}
  \label{prop:Zm=Zm'}
  If $\seq{m}$ and $\seq{m}'$ are two mutation
  sequences such that there is a frozen isomorphism between
  $\mu_{\seq{m}}(Q^{\wedge})$ and $\mu_{\seq{m}'}(Q^{\wedge})$, then we
  have
  $Z_{\seq{m}} = Z_{\seq{m}'}$.
\end{theorem}

First we will prove Theorem \ref{thm:main-1} in Section
\ref{sec:thm1-proof}.
Then the proof of Theorem \ref{prop:Zm=Zm'} will be given in Section
\ref{sec:thm2-proof}.

\subsection{Proof of Theorem \ref{thm:main-1}}
\label{sec:thm1-proof}
We first rewrite the product of quantum
dilogarithms in terms of $q$-binomials as follows:
\begin{prop}
  \label{prop:prodE-qbinom-rel} For any mutation sequence
  $\seq{m}=(m_{1},\dots, m_{T})$ and for any $n_1 ,\dots ,n_T \in \Z$,
  we have
  \begin{equation}
    \label{eq:prodE-qbinom-rel}
    \begin{split}
      \left(
      \EE(y^{\alpha_{1}};q^{\varepsilon_{1}})
      \cdots
      \EE(y^{\alpha_{T}};q^{\varepsilon_{T}})
      \right)^{-1}
      \left(
      \EE(q^{n_{1}} y^{\alpha_{1}};q^{\varepsilon_{1}})
      \cdots
      \EE(q^{n_{T}} y^{\alpha_{T}};q^{\varepsilon_{T}})
      \right) \\
      =
      \sum_{k\geq 0} \left( \prod_{t=1}^{T} q^{\frac{1}{2}\varepsilon_{t} k_{t}^{2}}
      \qb{\varepsilon_t ( n_t + \sum_{i=1}^{t-1}  \vev{\alpha_i , \alpha_t} k_i)}{k_t}_{q^{\varepsilon_t}} \right)
      y^{k_T \alpha_T} \cdots y^{k_1 \alpha_1} .
    \end{split}
  \end{equation}
\end{prop}
\begin{proof}
  A key is to express
  $q$-binomial coefficients as the ``ratio'' of $q$-dilogarithms:
 \begin{equation}
  \label{eq:binom-as-E-ratio}
  \begin{split}
    \EE(x;q)^{-1}\EE(q^{m}x;q) &=
    \begin{cases}
      \prod_{i=0}^{m-1}\left(1+q^{i+\frac{1}{2}}x\right) &(m \geq 0)  \\
      \prod_{i=0}^{-m-1}\left(1+q^{-i-\frac{1}{2}}x\right)^{-1} &(m<0)
    \end{cases}\\
    &=\sum_{k\geq 0}q^{\frac{1}{2}k^{2}}\qb{m}{k}_{q} x^{k}   
  \end{split}
 \end{equation}
  by using the $q$-binomial theorem. Moreover,
  by substituting $q^{-1}$ for $q$ and $-m$ for $m$, we get
  \begin{equation}
    \label{eq:binom-as-E-ratio-sign}
    \EE(x;q^{\varepsilon})^{-1}\EE(q^{m}x;q^{\varepsilon}) =
    \sum_{k\geq 0}q^{\frac{1}{2} \varepsilon k^{2}}\qb{\varepsilon m}{k}_{q^{\varepsilon}} x^{k}
  \end{equation}
  for $\varepsilon \in \{1, -1\}$.
  Then the proof is by induction on the length $T$ of the mutation sequence.
  For $T=0$, equality is obvious because both sides of \eqref{eq:prodE-qbinom-rel} are equal to $1$.
  For $T>0$, we set
  \begin{align*}
    A= \prod_{t=1}^{T-1} q^{\frac{1}{2}\varepsilon_{t} k_{t}^{2}}
    \qb{\varepsilon_t ( n_t + \sum_{i=1}^{t-1}  \vev{\alpha_i , \alpha_t} k_i)}{k_t}_{q^{\varepsilon_t}}.
  \end{align*}
  By induction hypothesis, we have
  \begin{align*}
    &\left(
    \EE(y^{\alpha_{1}};q^{\varepsilon_{1}})
    \cdots
    \EE(y^{\alpha_{T-1}};q^{\varepsilon_{T-1}})
    \right)^{-1}
    \left(
    \EE(q^{n_{1}} y^{\alpha_{1}};q^{\varepsilon_{1}})
    \cdots
    \EE(q^{n_{T-1}} y^{\alpha_{T-1}};q^{\varepsilon_{T-1}})
    \right) \\
   &    = \sum_{k_1,\dots k_{T-1}\geq 0} A y^{k_{T-1} \alpha_{T-1}} \cdots y^{k_1 \alpha_1} .
  \end{align*}
  Using the commutation relation
  \begin{align*}
    (y^{k_{T-1} \alpha_{T-1}} \cdots y^{k_1 \alpha_1} ) y^{\alpha_T} =
    q^{ \sum_{i=1}^{T-1} \vev{\alpha_i, \alpha_T}k_i} y^{\alpha_T}
    (y^{k_{T-1} \alpha_{T-1}} \cdots y^{k_1 \alpha_1} )
  \end{align*}
  and \eqref{eq:binom-as-E-ratio-sign}, we obtain
  \begin{align*}
    & \left(
    \EE(y^{\alpha_{1}};q^{\varepsilon_{1}})
    \cdots
    \EE(y^{\alpha_{T}};q^{\varepsilon_{T}})
    \right)^{-1}
    \left(
    \EE(q^{n_{1}} y^{\alpha_{1}};q^{\varepsilon_{1}})
    \cdots
    \EE(q^{n_{T}}y^{\alpha_{T}};q^{\varepsilon_{T}})
    \right) \\
    &=
    \EE(y^{\alpha_{T}};q^{\varepsilon_{T}})^{-1}
    \left(
    \sum_{k_1,\dots k_{T-1}\geq 0} A  y^{k_{T-1} \alpha_{T-1}} \cdots y^{k_1 \alpha_1}
    \right)
    \EE(q^{n_{T}}y^{\alpha_{T}};q^{\varepsilon_{T}}) \\
    &=
    \sum_{k_1,\dots k_{T-1}\geq 0}
    \EE(y^{\alpha_{T}};q^{\varepsilon_{T}})^{-1}
    \EE(q^{n_{T}+ \sum_{i=1}^{T-1} \vev{\alpha_i, \alpha_T}k_i}y^{\alpha_{T}};q^{\varepsilon_{T}})
    A  y^{k_{T-1} \alpha_{T-1}} \cdots y^{k_1 \alpha_1} \\
    &=
    \sum_{k \geq 0}
    q^{\frac{1}{2}\varepsilon_T k_T^2} \qb{\varepsilon_T (n_T + \sum_{i=1}^{T-1} \vev{\alpha_i, \alpha_T}k_i )}{k_T}_{q^{\varepsilon_T}}
    y^{k_T \alpha_T}
    A y^{k_T \alpha_{T-1}} \cdots y^{k_1 \alpha_1} \\
    &=
    \sum_{k\geq 0} \left( \prod_{t=1}^{T} q^{\frac{1}{2}\varepsilon_{t} k_{t}^{2}}
    \qb{\varepsilon_t(n_t + \sum_{i=1}^{t-1}  \vev{\alpha_i , \alpha_t} k_i )}{k_t}_{q^{\varepsilon_t}} \right)
    y^{k_T \alpha_T} \cdots y^{k_1 \alpha_1} .
  \end{align*}
  This completes the proof of Proposition \ref{prop:prodE-qbinom-rel}.
\end{proof}

We introduce the notion of a state vector considered in \cite{KT2}.  The
\emph{state vector} of $Q(t)$ is defined by
\begin{equation}
 \label{eq:psi-def} \psi(t):=\sum_{i=1}^{n} s_{i}(t) c_{i}(t) ~\in
 \Z^{n}\qquad (0\leq t\leq T).
\end{equation}
By the definition of $c$-vectors,
we have $\psi(0)=(s_1(0) ,\dots , s_n(0) ) = r$.

We use the following properties shown in \cite{KT2}
\footnote{Note that $k_{t}^{\vee}$ in this paper is $-k_{t}^{\vee}$ in \cite{KT2}.}:
\begin{itemize}
\item
For the mutation sequence \eqref{eq:Q-seq-2}, we have
 \begin{equation}
  \label{eq:Bcc-t}
  B(t)_{ij}=\vev{c_{i}(t),c_{j}(t)},\qquad (0\leq t\leq T)
 \end{equation}
 or equivalently,
 \begin{equation}
  \label{eq:CBC=B}
   C(t)B(0)C(t)^{\top} =B(t).\qquad  (0\leq t\leq T)
 \end{equation}

\item
 Along the mutation sequence \eqref{eq:Q-seq-2},
 the state vector changes as
 \begin{equation}
  \label{eq:psi-change}
   \psi(t)=
   \psi(t{-}1)   -k_{t} \alpha_{t},\qquad (t=1,\dots,T).
 \end{equation}

\item
The state vectors of the initial and the final
 quivers are related as
 \begin{equation}
    \label{eq:psi-in-out}
     \psi(0)-\psi(T)=\sum_{t=1}^{T}k_{t}\alpha_{t}.
 \end{equation}

\item
For any mutation sequence, we have
 \begin{equation}
  \label{eq:kkv-k2-rel}
   \rule[-6pt]{0pt}{20pt}
   -\sum_{t=1}^{T}\varepsilon_{t}k_{t}k_{t}^{\vee}
   +\vev{\psi(0), \psi(T)}
   =  \sum_{t=1}^{T}\varepsilon_{t}k_{t}^{2}
   - \sum_{1\leq i< j\leq
   T}k_{i}k_{j}\vev{\alpha_{i},\alpha_{j}}.
 \end{equation}

\end{itemize}

We also need the following:
\begin{prop}
  \label{prop:kkv-alpha-rel}
For any mutation sequences, we have
\begin{equation}
  \label{eq:kkv-alpha-rel}
  k_t + k_t^{\vee} = \varepsilon_t
  (\vev{\alpha_t, \psi(0)} + \sum_{i=1}^{t-1} \vev{\alpha_i , \alpha_t} k_i ).
\end{equation}
\end{prop}
\begin{proof}
  \begin{align*}
    k_t + k_t^{\vee} &
      =\sum_{i=1}^n B(t-1)_{m_t,i} s_i(t-1)
      &(\text{by \eqref{eq:k-kv-rel}}) \\
      &=\sum_{i=1}^n \vev{c_{m_t}(t-1),c_{i}(t-1)} s_i(t-1)
      &(\text{by \eqref{eq:Bcc-t}}) \\
      &=\vev{c_{m_t}(t-1),\sum_{i=1}^n c_i(t-1) s_i(t-1)} \\
      &=\vev{c_{m_t}(t-1), \psi(t-1)}
      &(\text{by \eqref{eq:psi-def}})\\
      &=\vev{\varepsilon_t \alpha_t, \psi(t-1)}
      &(\text{by \eqref{eq:alpha-def}})\\
      &=\vev{\varepsilon_t \alpha_t, \psi(0) - \sum_{i=1}^{t-1} \alpha_i k_i}
      &(\text{by \eqref{eq:psi-change}}) \\
      &=\varepsilon_t ( \vev{\alpha_t, \psi(0)} + \sum_{i=1}^{t-1} \vev{\alpha_i , \alpha_t}k_i) .
      &(\text{by skewness of $\vev{~,~}$})
  \end{align*}
\end{proof}

We are now ready to prove Theorem \ref{thm:main-1}.
The partition function associated with the mutation sequence $\seq{m}$ and initial $s$-variables
$r=(r_1 , \dots , r_n)$ is given by
\begin{align*}
  Z_{\seq{m}}(r) &=\sum_{k\geq 0} \biggl(
  \prod_{t=1}^{T}
  W^{\varepsilon_{t}}(k_{t},k_{t}^{\vee})
  \biggr)\;
  y^{\sum_{t=1}^{T}k_{t}\alpha_{t}} \\
  &= \sum_{k\geq 0} \biggl(
  \prod_{t=1}^{T}
  q^{-\frac{1}{2} \varepsilon_t k_t k_t^{\vee} }
  \qb{k_{t}+k_{t}^{\vee}}{k_t}_{q^{\varepsilon_t}}
  \biggr)\;
  y^{\sum_{t=1}^{T}k_{t}\alpha_{t}}.
\end{align*}
On the other hand, the ratio of the quantum dilogarithms products along $\seq{m}$ is given by
\begin{align*}
  &\left(
  \EE(y^{\alpha_{1}};q^{\varepsilon_{1}})
  \cdots
  \EE(y^{\alpha_{T}};q^{\varepsilon_{T}})
\right)^{-1}
\EE(q^{\vev{\alpha_{1},r}} y^{\alpha_{1}};q^{\varepsilon_{1}})
\cdots
\EE(q^{\vev{\alpha_{T},r}}y^{\alpha_{T}};q^{\varepsilon_{T}}) \\
&=\sum_{k\geq 0} \left( \prod_{t=1}^{T} q^{\frac{1}{2}\varepsilon_{t} k_{t}^{2}}
\qb{\varepsilon_t ( \vev{\alpha_t , r} + \sum_{i=1}^{t-1}  \vev{\alpha_i , \alpha_t} k_i)}{k_t}_{q^{\varepsilon_t}} \right)
y^{k_T \alpha_T} \cdots y^{k_1 \alpha_1} \\
&=\sum_{k\geq 0} \left( \prod_{t=1}^{T} q^{\frac{1}{2}\varepsilon_{t} k_{t}^{2}}
\qb{\varepsilon_t ( \vev{\alpha_t , \psi(0)} + \sum_{i=1}^{t-1}  \vev{\alpha_i , \alpha_t} k_i)}{k_t}_{q^{\varepsilon_t}} \right)
y^{k_T \alpha_T} \cdots y^{k_1 \alpha_1} \\
&=\sum_{k\geq 0} \left(\prod_{t=1}^{T} q^{\frac{1}{2} \varepsilon_{t} k_{t}^{2}}
\qb{\varepsilon_t ( \vev{\alpha_t , \psi(0)} + \sum_{i=1}^{t-1}  \vev{\alpha_i , \alpha_t} k_i)}{k_t}_{q^{\varepsilon_t}} \right)
q^{-\frac{1}{2} \sum_{1\leq i \leq j \leq T} k_i k_j \vev{\alpha_i , \alpha_j} }y^{\sum_{t=1}^T k_t \alpha_t}.
\end{align*}
Therefore, to prove Theorem \ref{thm:main-1}, it suffices to show that the
following equalities:
\begin{align}
  \label{eq:param of binom}
  &k_t + k_t^{\vee} = \varepsilon_t ( \vev{\alpha_t , \psi(0)} + \sum_{i=1}^{t-1}  \vev{\alpha_i , \alpha_t} k_i),
\end{align}
\begin{align}
  \label{eq:deg of q}
  &-\frac{1}{2} \sum_{t=1}^T \varepsilon_t  k_t k_t^{\vee}
  + \frac{1}{2} \vev{\sum_{t=1}^T k_t \alpha_t , r}
  = \frac{1}{2} \sum_{t=1}^T \varepsilon_t k_t^2
  -\frac{1}{2} \sum_{1\leq i \leq j \leq T} k_i k_j \vev{\alpha_i , \alpha_j}.
\end{align}
The equality \eqref{eq:param of binom} follows immediately from
Proposition \ref{prop:kkv-alpha-rel}.  \eqref{eq:deg of q} follows from
\eqref{eq:kkv-k2-rel} because
\begin{align*}
  \Bvev{\sum_{t=1}^T k_t \alpha_t , r} =
  \vev{\psi(0)-\psi(T),\psi(0)}=\vev{\psi(0),\psi(T)}
\end{align*}
by \eqref{eq:psi-in-out} and the skewness of $\vev{~,~}$.
This completes the proof of Theorem \ref{thm:main-1}.

\subsection{Proof of Theorem \ref{prop:Zm=Zm'}}
\label{sec:thm2-proof} 
We first review some known results about quantum dilogarithms. For a mutation sequence
$\seq{m}=(m_1 ,\dots ,m_{T})$ of $Q$ we consider the following ordered
product of quantum dilogarithms:
\begin{equation}
  \EE(Q;\seq{m}):=
  \EE(y^{\alpha_{1}};q^{\varepsilon_{1}})
  \EE(y^{\alpha_{2}};q^{\varepsilon_{2}})
  \cdots
  \EE(y^{\alpha_{T}};q^{\varepsilon_{T}})\quad \in \Ah_{Q}.
\end{equation}

\begin{theorem}[Keller\cite{Keller2012}, Nagao\cite{Nagao2011}]
 \label{thm:EEm=EEm'} If $\seq{m}$ and $\seq{m}'$ are two mutation
 sequences such that there is a frozen isomorphism between
 $\mu_{\seq{m}}(Q^{\wedge})$ and $\mu_{\seq{m}'}(Q^{\wedge})$, then we
 have $\EE(Q;\seq{m})=\EE(Q;\seq{m}')$.
\end{theorem}

\begin{theorem}[Keller\cite{Keller2012}]
 \label{thm:EEm-reddening}
 If $\seq{m}$ and $\seq{m}'$ are reddening
 sequences on the quiver $Q$, then there is a frozen isomorphism between
 the final ice quivers $\mu_{\seq{m}}(Q^{\wedge})\simeq
 \mu_{\seq{m}'}(Q^{\wedge})$.
\end{theorem}

Theorems \ref{thm:EEm=EEm'} and \ref{thm:EEm-reddening} imply that if
$Q$ admits a reddening sequence $\seq{m}$, then the power series
\begin{equation}
  \EE_{Q}:= \EE(Q;\seq{m})\in \Ah_{Q}.
\end{equation}
is independent of the choice of the reddening sequence $\seq{m}$.
Thus it is intrinsically associated with the quiver $Q$.
Keller \cite{Keller2013a} named this invariant as \emph{combinatorial Donaldson-Thomas (DT) invariant}.
Note that the statements of Theorems \ref{thm:EEm-reddening} is
combinatorial, but the known proofs are
based on categorification in terms of Ginzburg dg-algebra
\cite{Ginzburg2006}.
The well-known pentagon identity
\begin{align*}
  \EE(x)\EE(y)= \EE(y)\EE(q^{-\frac{1}{2}}xy)\EE(x)
\end{align*}
for $xy = q yx$ is nothing but the combinatorial
DT invariant of $A_{2}$ quiver $Q=(1{\to}2)$ corresponding to the two
reddening sequences $\seq{m}=(1,2)$ and $\seq{m}'=(2,1,2)$ depicted in
Figure \ref{fig:pentagon}.

Thanks to the identities of quantum dilogarithms, 
Theorem \ref{prop:Zm=Zm'} follows immediately from
our first main result Theorem \ref{thm:main-1}.
\begin{proof}[Proof of Theorem \ref{prop:Zm=Zm'}]
From Theorem \ref{thm:EEm=EEm'}, we have
$\EE(Q;\seq{m})=\EE(Q;\seq{m}')$.
This equality also holds if
we replace $(y^{\alpha_1}, \dots , y^{\alpha_T})$ by $(q^{n_1}
y^{\alpha_1}, \dots ,q^{n_T} y^{\alpha_T})$ for $(n_1 , \dots , n_T) \in
\Z^T$ because they have the same commutation relations.  Therefore,
Theorem \ref{thm:main-1} implies that $Z_{\seq{m}} = Z_{\seq{m}'}$.
\end{proof}

One of the consequences is that, from a viewpoint of Theorem \ref{thm:EEm-reddening},
if a quiver $Q$ admit a reddening sequence $\seq{m}$,
our partition function $Z_{\seq{m}}$ is independent of the choice of the reddening sequence.

\section{Examples}
\label{sec:examples}

In this section, we show some sample computation of partition functions
for some mutation sequences. We hope that these examples illustrate how
our results provide a systematic way of constructing various
$q$-binomial multisum identities.

Throughout this section, we simply write $W(k_t,k_t^{\vee}) $ for
$W^{+}(k_t,k_t^{\vee})$ when we consider green sequences.
\begin{example}
  Consider the $A_2$ quiver
  \begin{equation}
    \label{eq:A2 quiver}
    Q=\vcenter{
      \xymatrix @R=6mm @C=6mm @M=4pt{
        1 \ar[r] & 2 .}
    }
  \end{equation}
  and a mutation sequence $\seq{m} = (1,2)$.
  The partition function of $\seq{m}$ is computed in Example \ref{example:A2-part-poly}:
  \begin{align}
    \label{eq:Zm A2}
    Z_{\seq{m}}(r_1,r_2) = \sum_{k_1,k_2\geq 0} W(k_1,-k_1+r_2)W(k_2,k_1-k_2-r_1) y^{(k_1,k_2)} .
  \end{align}
  The coefficient of $y^{(\beta_1,\beta_2)}$
  in the partition function \eqref{eq:Zm A2} is given by
 \begin{equation}
    \label{eq:coeff-A2-Zm}
     \begin{split}
    [y^{(\beta_1,\beta_2)}] Z_{\seq{m}}(r_1,r_2)
    &=W(\beta_1,-\beta_1+r_2)W(\beta_2,\beta_1-\beta_2-r_1) \\
    &= q^{\frac{1}{2}(\beta_1^2+\beta_2^2-\beta_1\beta_2-\beta_1r_2+\beta_2r_1)}
    \qb{r_2}{\beta_1}_q \qb{\beta_1-r_1}{\beta_2}_q .
     \end{split}
 \end{equation}
  By Theorem \ref{thm:main-1}, this is equal to
  \begin{align*}
    q^{-\frac{1}{2}(\beta_1 r_2 - \beta_2 r_1)}
    [y^{(\beta_1,\beta_2)}]
    \left((\EE(y_1;q) \EE(y_2;q))^{-1} \EE(q^{r_2} y_1;q) \EE(q^{-r_1} y_2;q)\right).
  \end{align*}
\end{example}

\begin{example}
  We again consider the $A_2$ quiver \eqref{eq:A2 quiver}.
  The mutation sequences $\seq{m} = (1,2)$ and $\seq{m}' = (2,1,2)$ of $Q$ are maximal green sequences.
  In addition, the permutation of vertices $(1,2)$ is a frozen isomorphism between
  $\mu_{\seq{m}}(Q^\vee)$ and $\mu_{\seq{m}'}(Q^{\vee})$(see Figure \ref{fig:pentagon}).
  It follows from Theorem \ref{prop:Zm=Zm'} that
  \begin{align} \label{eq:A_2 m=m'}
    Z_{\seq{m}}(r_1,r_2)=Z_{\seq{m}'}(r_1,r_2)
  \end{align}
  for all $r_1,r_2 \in \Z$.
  Let us express \eqref{eq:A_2 m=m'} in terms of mutation weight.

  In the mutation sequences $\seq{m'}$, the $s$-variables change as follows:
  \begin{align*}
    \begin{array}{c|c|c}
      & s_1(t)                   & s_2(t)                      \\ \hline
      t=0 & r_1               & r_2                  \\
      t=1 & r_1               & k_1-r_2+r_1  \\
      t=2 & k_2-r_1+s_2(1)   & k_1-r_2+r_1  \\
      t=3 & k_2-r_1+s_2(1)   & k_3-s_2(2)+s_1(2)  \\
    \end{array}
  \end{align*}
  The $k^{\vee}$-variables are computed to be
  \begin{equation}
    \begin{split}
      k_1^{\vee} &=-s_2(0) -s_2(1)  =-k_1-r_1,  \\
      k_2^{\vee} &=-s_1(1) -s_1(2)  =-k_1-k_2-r_1+r_2, \\
      k_3^{\vee} &=-s_2(2) -s_2(3) =-k_1-k_2-k_3+r_2. \\
    \end{split}
  \end{equation}
  The partition function of $\seq{m}'$ is given by
  \begin{equation}
    \begin{split}
      \label{eq:Zm' A2}
      Z_{\seq{m}'} (r_1,r_2)
      &=\sum_{k_1,k_2,k_3 \geq 0}  W(k_1,-k_1-r_1)W(k_2,-k_1-k_2-r_1+r_2)\\
      &\cdot W(k_3,-k_1-k_2-k_3+r_2) y^{(k_2+k_3,k_1+k_2)} .
    \end{split}
  \end{equation}

  Comparing the coefficients of $y^{(\beta_1,\beta_2)}$ in \eqref{eq:Zm A2} and \eqref{eq:Zm' A2}
  gives a five-term identity of mutation weights:
  \begin{equation}
    \begin{split}
      &W(\beta_1,-\beta_1+r_2)W(\beta_2,\beta_1-\beta_2+r_1) \\
      &=\sum_{\substack{k_1,k_2,k_3 \geq 0\\k_2+k_3=\beta_1 \\ k_1+k_2=\beta_2}}
      W(k_1, -k_1-r_1)W(k_2,-\beta_2-r_1+r_2) W(k_3,-\beta_2-k_3+r_2).
    \end{split}
  \end{equation}
  Multiplying both sides by $q^{-\frac{1}{2}(\beta_1^2+\beta_2^2-\beta_1\beta_2
  -\beta_1r_2+\beta_2r_1)}$, we get a five-term $q$-binomial identity:
  \begin{align}
    \qb{r_2}{\beta_1}_q \qb{\beta_1-r_1}{\beta_2}_q
    =\sum_{\substack{k_1,k_2,k_3 \geq 0\\k_2+k_3=\beta_1 \\ k_1+k_2=\beta_2}}
    q^{k_1 k_3}\qb{-r_1}{k_1}_q \qb{-r_1+r_2-k_1}{k_2}_q \qb{-\beta_2+r_2}{k_3}_q .
  \end{align}

\begin{rem}
This identity is equivalent to the Stanley's
identity (see \cite{Gould1972})
\begin{equation}
 \begin{split}
 \qb{c+a}{a}_{q}
 \qb{d+b}{b}_{q}
 =
  \sum_{k=0}^{\min(a,b)}
  q^{(a-k)(b-k)}
   \qb{c+d+k}{k}_{q}
   \qb{c+a-b}{a-k}_{q} \qb{d+b-a}{b-k}_{q}.
 \end{split}
\end{equation}
\end{rem}
\end{example}

\begin{example}
  Consider an alternating quiver of type $A_3$
  \begin{equation}\label{eq:A3-quiver}
    Q=\vcenter{
      \xymatrix @R=6mm @C=6mm @M=4pt{
        1 \ar[r] & 2  & 3\ar[l]}
    }.
  \end{equation}

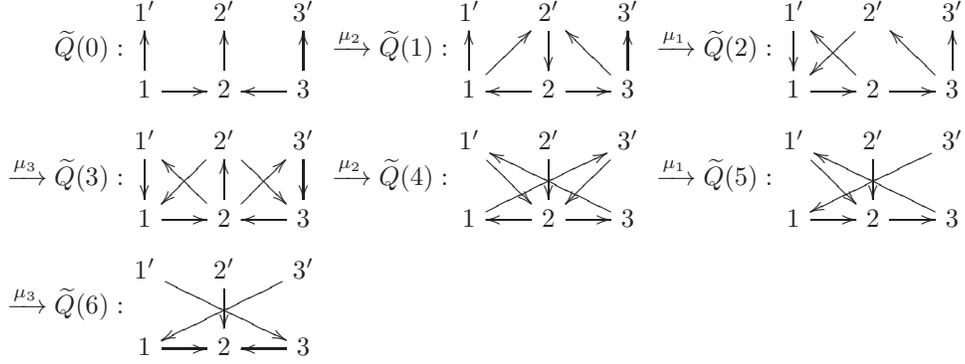
\begin{figure}
  \label{fig:A3 m'}
\begin{align*}
  \widetilde{Q}(0):\hspace{-1mm}
  \vcenter{
    \xymatrix @R=5mm @C=5mm @M=4pt{
      1' & 2'  & 3'  \\
      1 \ar[r] \ar[u] & 2 \ar[u]  & 3 \ar[l] \ar[u] }}
  &\xrightarrow{\mu_2}
  \widetilde{Q}(1):\hspace{-1mm}
  \vcenter{
    \xymatrix @R=5mm @C=5mm @M=4pt{
      1'  & 2' \ar[d] & 3' \\
      1 \ar[u] \ar[ru] & 2 \ar[l] \ar[r]  & 3 \ar[u] \ar[lu]}}
  \xrightarrow{\mu_1}
  \widetilde{Q}(2):\hspace{-1mm}
  \vcenter{
    \xymatrix @R=5mm @C=5mm @M=4pt{
      1' \ar[d] & 2' \ar[ld]  & 3'   \\
      1 \ar[r]  & 2 \ar[lu] \ar[r]  & 3 \ar[lu] \ar[u] }}\\
  \xrightarrow{\mu_3}
  \widetilde{Q}(3):\hspace{-1mm}
  \vcenter{
    \xymatrix @R=5mm @C=5mm @M=4pt{
      1' \ar[d]  & 2' \ar[ld]  \ar[rd] & 3' \ar[d] \\
      1  \ar[r]  & 2 \ar[ru] \ar[lu] \ar[u] & 3 \ar[l]}}
  &\xrightarrow{\mu_2}
  \widetilde{Q}(4):\hspace{-1mm}
  \vcenter{
    \xymatrix @R=5mm @C=5mm @M=4pt{
      1' \ar[rd] & 2' \ar[d] & 3' \ar[ld] \\
      1 \ar[rru] & 2 \ar[l] \ar[r]  & 3 \ar[llu]}}
  \xrightarrow{\mu_1}
  \widetilde{Q}(5):\hspace{-1mm}
  \vcenter{
    \xymatrix @R=5mm @C=5mm @M=4pt{
      1' \ar[rd] & 2' \ar[d]  & 3' \ar[lld]  \\
      1 \ar[r]  & 2 \ar[r]  & 3 \ar[llu]  }}\\
  \xrightarrow{\mu_3}
  \widetilde{Q}(6):\hspace{-1mm}
  \vcenter{
    \xymatrix @R=5mm @C=5mm @M=4pt{
      1' \ar[rrd] & 2' \ar[d]  & 3' \ar[lld]  \\
      1 \ar[r]  & 2   & 3 \ar[l]  }}
  \end{align*}
      \caption{A mutation sequence $\seq{m}' = (2,1,3,2,1,3)$ for an $A_3$ quiver.}
\end{figure}

  The mutation sequences $\seq{m} = (1,3,2)$ and $\seq{m}' = (2,1,3,2,1,3)$ of $Q$ are maximal green sequences.
  In addition, the permutation of vertices $(1,3)$ is a frozen isomorphism between $\mu_{\seq{m}}(Q^\vee)$ and $\mu_{\seq{m}'}(Q^{\vee})$
  (see Figure \ref{fig:A3 m'}).
  It follows from Theorem \ref{prop:Zm=Zm'} that
  \begin{align} \label{eq:A_3 m=m'}
    Z_{\seq{m}}(r_1,r_2,r_3)=Z_{\seq{m}'}(r_1,r_2,r_3).
  \end{align}
  for all $r_1,r_2,r_3 \in \Z$.
  Let us write \eqref{eq:A_3 m=m'} in terms of mutation weights.

  In the mutation sequences $\seq{m}$, the $s$-variables change as follows:
  \begin{align*}
    \begin{array}{c|c|c|c}
      & s_1(t)                   & s_2(t)               & s_3(t) \\ \hline
      t=0 & r_1               & r_2           & r_3 \\
      t=1 & k_1-r_1       & r_2           & r_3 \\
      t=2 & k_1-r_1       & r_2           & k_2-r_3 \\
      t=3 & k_1-r_1       & k_3-r_2   & k_2-r_3
    \end{array}
  \end{align*}
  The $k^{\vee}$-variables are computed to be
  \begin{equation}
    \begin{split}
      k_1^{\vee} &=-s_1(0) - s_1(1) + s_2(0)     =-k_1 +r_2,  \\
      k_2^{\vee} &=-s_3(1) - s_3(2) + s_2(1)      =-k_2 +r_2, \\
      k_3^{\vee} &=-s_2(2) - s_2(3) +s_1(2) +s_3(2)  =k_1+k_2-k_3-r_1-r_3 . \\
    \end{split}
  \end{equation}
  The partition function of $\seq{m}$ is given by
  \begin{equation}
    \begin{split}
      &Z_{\seq{m}}(r_1,r_2,r_3)
      =\sum_{k_1,k_2,k_3 \geq 0}
      W(k_1,-k_1+r_2) \\
      &\cdot W(k_2,-k_2+r_2)W(k_3,-k_3+k_1+k_2-r_1-r_3)y^{(k_1,k_3,k_2)} .
    \end{split}
  \end{equation}

  On the other hand, in the mutation sequences $\seq{m}'$,
  the $s$-variables change as follows:
  \begin{align*}
    \begin{array}{c|c|c|c}
      & s_1(t)                   & s_2(t)                      & s_3(t) \\ \hline
      t=0 & r_1                & r_2                  & r_3 \\
      t=1 & r_1                & k_1-r_2+r_1+r_3  & r_3 \\
      t=2 & k_2-r_1+s_2(1)     & k_1-r_2+r_1+r_3  & r_3 \\
      t=3 & k_2-r_1+s_2(1)     & k_1-r_2+r_1+r_3  & k_3-r_3+s_2(2) \\
      t=4 & k_2-r_1+s_2(1)     & k_4-s_2(3)+s_1(3)+s_3(3) & k_3-r_3+s_2(2) \\
      t=5 & k_5-s_1(4)+s_2(4)  & k_4-s_2(3)+s_1(3)+s_3(3) & k_3-r_3+s_2(2) \\
      t=6 & k_5-s_1(4)+s_2(4)  & k_4-s_2(3)+s_1(3)+s_3(3) & k_6-s_3+s_2(5)
    \end{array}
  \end{align*}
  The $k^{\vee}$-variables are computed to be
  \begin{equation}
    \begin{split}
      k_1^{\vee} &=-s_2(0) -s_2(1)  =-k_1-r_1-r_3 , \\
      k_2^{\vee} &=-s_1(1) -s_1(2)  =-k_1-k_2-r_1+r_2-r_3 ,\\
      k_3^{\vee} &=-s_3(2) -s_3(3)  =-k_1-k_3-r_1+r_2-r_3 ,\\
      k_4^{\vee} &=-s_2(3) -s_2(4) =-2k_1-k_2-k_3-k_4-r_1+2r_2-r_3 ,\\
      k_5^{\vee} &=-s_1(4) -s_1(5) =-k_1-k_2-k_3-k_4-k_5+r_2 ,\\
      k_6^{\vee} &=-s_3(5) -s_3(6) =-k_1-k_2-k_3-k_4-k_6+r_2 .
    \end{split}
  \end{equation}
  The partition function of $\seq{m}'$ is given by
  \begin{equation}
    \begin{split}
      \label{eq:Zm' A3}
      &Z_{\seq{m}'}  (r_1,r_2,r_3)
      =\sum_{k_1,\dots,k_6 \geq 0}  W(k_1,-k_1-r_1-r_3)W(k_2,-k_1-k_2-r_1+r_2-r_3)\\
      &\cdot W(k_3,-k_1-k_3-r_1+r_2-r_3)  W(k_4,-2k_1-k_2-k_3-k_4-r_1+2r_2-r_3)\\
      &\cdot W(k_5,-k_1-k_2-k_3-k_4-k_5+r_2)  W(k_6,-k_1-k_2-k_3-k_4-k_6+r_2) \\
      &\cdot y^{(k_2+k_4+k_6, k_1+k_2+k_3+k_4, k_3+k_4+k_5)} .
    \end{split}
  \end{equation}

  Comparing the coefficients of $y^{(\beta_1,\beta_2,\beta_3)}$ in $Z_{\seq{m}}$ and
  $Z_{\seq{m}'}$ gives a nine-term identity of mutation weights:
  \begin{equation}
    \begin{split}
      &W(\beta_1,-\beta_1+r_2)W(\beta_3,-\beta_3+r_2)W(\beta_2,\beta_1-\beta_2+\beta_3-r_1-r_3) \\
      &=\sum_{k_1,\dots,k_6 \geq 0 , (*)}
      W( k_1, -k_1-r_1-r_3)W(k_2,-k_1-k_2-r_1+r_2-r_3) \\
      &\quad \cdot W(k_3,-k_1-k_3-r_1+r_2-r_3) W( k_4, -k_1-\beta_2-r_1+2r_2-r_3)\\
      &\quad \cdot W(k_5,-k_5-\beta_2+r_2) W(k_6,-k_6-\beta_2+r_2),
    \end{split}
  \end{equation}
  where $(*)$ is the following conditions:
  \begin{align}
  (*) &\Leftrightarrow
  \begin{cases}
    k_2+k_4+k_6=\beta_1 \\
    k_1+k_2+k_3+k_4=\beta_2 \\
    k_3+k_4+k_5=\beta_3
  \end{cases} .
  \end{align}
  This gives the following identity of $q$-binomial coefficients:
  \begin{equation}
    \begin{split}
      \label{eq:A3 qbinom identity}
      &\qb{r_2}{\beta_1}_q \qb{r_2}{\beta_3}_q \qb{\beta_1+\beta_3-r_1-r_3}{\beta_2}_q \\
      &=\sum_{k_1,\dots,k_6 \geq 0 , (*)}
      q^{k_1k_4+k_1k_5+k_1k_6+k_2k_5+k_3k_6}
      \qb{-r_1-r_3}{k_1}_q 
     \\
      &\qquad\cdot\qb{-k_1-r_1+r_2-r_3}{k_2}_q \qb{-k_1-r_1+r_2-r_3}{k_3}_q \\
      &\qquad\cdot\qb{-k_1+k_4-\beta_2-r_1+2r_2-r_3}{k_4}_q \qb{-\beta_2+r_2}{k_5}_q \qb{-\beta_2+r_2}{k_6}_q
    \end{split}
  \end{equation}

  For example, if we set $(\beta_1,\beta_2,\beta_3)=(1,1,2)$ and $(r_1,r_2,r_3)=(0,6,-2)$,
  \eqref{eq:A3 qbinom identity} gives
  \begin{equation*}
    \begin{split}
      \qb{6}{1}_{q} \qb{6}{2}_{q} \qb{5}{1}_{q}
       = &\qb{2}{0}_q \qb{8}{0}_q \qb{8}{0}_q \qb{14}{1}_q \qb{5}{1}_q \qb{5}{0}_q
      +q \qb{2}{0}_q \qb{8}{0}_q \qb{8}{1}_q \qb{13}{0}_q \qb{5}{1}_q \qb{5}{1}_q \\
      &+ q^2 \qb{2}{0}_q \qb{8}{1}_q \qb{8}{0}_q \qb{13}{0}_q \qb{5}{2}_q \qb{5}{0}_q
      + q^3 \qb{2}{1}_q \qb{7}{0}_q \qb{7}{0}_q \qb{12}{0}_q \qb{5}{2}_q \qb{5}{1}_q .
    \end{split}
  \end{equation*}
  In fact, we can check both sides are equal to
  \begin{align*}
    &1 + 3q +7q^2 + 13q^3 + 22q^4 + 32q^5 + 42q^6 + 50q^7 + 55q^8 + 55q^9 \\
    &+50q^{10} + 42q^{11} + 32q^{12} + 22q^{13} + 13q^{14} + 7q^{15}+3q^{16} + q^{17}.
  \end{align*}

\end{example}

\begin{example}
  \label{example:B_2}
  We give a example which has negative sign mutations.
  Let $Q$ be a quiver given by
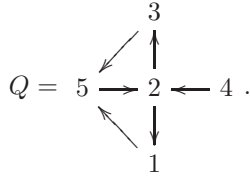
\begin{align*}
  Q=
  \vcenter{
    \xymatrix @R=5mm @C=5mm @M=4pt{
      & 3 \ar[ld]  &    \\
      5 \ar[r]  & 2 \ar[u] \ar[d]  & 4 \ar[l]  \\
      & 1 \ar[lu] & }}.
\end{align*}
We can check directly that both of two mutation sequences
\begin{align*}
  \seq{m}  &= (1,3,4,2,1,3,5,2), \\
  \seq{m}' &= (2,1,3,5,2,1,3,4,2,1,3,5)
\end{align*}
are reddening sequences.
The signs $\varepsilon=(\varepsilon_1^\vee ,\dots,\varepsilon_8^\vee)$ of $\seq{m}$ and
$\varepsilon'=({\varepsilon'}_1^{\vee},\dots, {\varepsilon'}_{12}^{\vee})$ of $\seq{m}'$ are
\begin{align*}
  \varepsilon&=(+,+,+,+,+,+,+,+), \\
  \varepsilon'&=(+,+,+,+,+,-,-,+,+,+,+,+).
\end{align*}
The mutation sequence $\seq{m}'$ is a example of
a reddening sequence that is not a maximal green sequence.

Furthermore,
\begin{align*}
  \left(
  \begin{array}{ccccc}
    1 & 2 & 3 & 4 & 5 \\
    3 & 2 & 1 & 5 & 4
  \end{array}
  \right)
  \in S_5
\end{align*}
gives a frozen isomorphism between $\mu_{\seq{m}} (Q^\wedge)$ and $\mu_{\seq{m}'} (Q^\wedge )$.
These mutation sequences appear in periodicity of a level $2$ restricted $T$ and $Y$-system of $B_2$ \cite{IIKKN2013}.
We get
\begin{align}
  \label{Zm=Zm' B2}
  Z_{\seq{m}} (r_1, \dots ,r_5) =Z_{\seq{m}'} (r_1, \dots ,r_5)
\end{align}
by Theorem \ref{prop:Zm=Zm'}.
Comparing the coefficients of $y^{(\beta_1,\dots,\beta_5)}$ of both sides of \eqref{Zm=Zm' B2}
gives a twenty terms identity of the mutation weights
\begin{equation}
  \begin{split}
    \sum_{k\geq 0, (*)}
    W^{\varepsilon_1}(k_1,k_1^\vee) \cdots
    W^{\varepsilon_8}(k_8,k_8^\vee)
    =\sum_{k' \geq 0, (*)'}
    W^{\varepsilon'_1}(k'_1,{k'}_1^{\vee}) \cdots
    W^{\varepsilon'_{12}}(k'_{12},{k'}_{12}^{\vee}),
  \end{split}
\end{equation}
where the $k^{\vee}$-variables $k_1^\vee ,\dots,k_8^\vee$ and ${k'}_1^{\vee},\dots, {k'}_{12}^{\vee}$
are computed to be
\begin{equation}
  \begin{split}
    &k_1^\vee=-k_{1} - r_{2} + r_{5}, \\
    &k_2^\vee=-k_{2} - r_{2} + r_{5}, \\
    &k_3^\vee=-k_{3} + r_{2}, \\
    &k_4^\vee=-k_{1} - k_{2} + k_{3} - k_{4} + r_{1} - 2 r_{2} + r_{3} - r_{4} + r_{5}, \\
    &k_5^\vee=-k_{1} - k_{2} + k_{3} - k_{4} - k_{5} + r_{1} - r_{2} + r_{3} - r_{4}, \\
    &k_6^\vee=-k_{1} - k_{2} + k_{3} - k_{4} - k_{6} + r_{1} - r_{2} + r_{3} - r_{4}, \\
    &k_7^\vee=k_{1} + k_{2} + k_{4} - k_{7} - r_{1} + r_{2} - r_{3}, \\
    &k_8^\vee=-k_{1} - k_{2} + k_{3} - 2 k_{4} - k_{5} - k_{6} + k_{7} - k_{8}
    + r_{1} + r_{3} - r_{4} - r_{5},
  \end{split}
\end{equation}
and
\begin{equation}
  \begin{split}
    &{k'}_1^{\vee}=-k'_{1} + r_{1} + r_{3} - r_{4} - r_{5}, \\
    &{k'}_2^{\vee}=k'_{1} - k'_{2} - r_{2} + r_{5}, \\
    &{k'}_3^{\vee}=k'_{1} - k'_{3} - r_{2} + r_{5}, \\
    &{k'}_4^{\vee}=-k'_{1} - k'_{4} + r_{2} - r_{4} - r_{5}, \\
    &{k'}_5^{\vee}=-k'_{1} + k'_{2} + k'_{3} - k'_{4} - k'_{5} - r_{1} + r_{2} - r_{3}, \\
    &{k'}_6^{\vee}=-k'_{1} + k'_{5} - k'_{6} + r_{2} - r_{5}, \\
    &{k'}_7^{\vee}=-k'_{1} + k'_{5} - k'_{7} + r_{2} - r_{5},  \\
    &{k'}_8^{\vee}=-k'_{4} - k'_{5} - k'_{8} - r_{4} + r_{5}, \\
    &{k'}_9^{\vee}=-k'_{1} - k'_{2} - k'_{3} + k'_{5} - k'_{6} - k'_{7} - k'_{8} - k'_{9}
    + r_{1} + r_{2} + r_{3} - r_{4} - r_{5}, \\
    &{k'}_{10}^{\vee}=k'_{1} - k'_{5} + k'_{9} - k'_{10} - r_{2} + r_{5}, \\
    &{k'}_{11}^{\vee}=k'_{1} - k'_{5} + k'_{9} - k'_{11}  - r_{2} + r_{5}, \\
    &{k'}_{12}^{\vee}=-k'_{1}- k'_{4} - k'_{8} - k'_{9}  - k'_{12} + r_{2}.
  \end{split}
\end{equation}
In addition, $(*)$ and $(*)'$ are the following conditions:
\begin{align}
  (*) &\Leftrightarrow
  \begin{cases}
    k_1 + k_4 + k_6=\beta_1 \\
    k_4 + k_5 + k_6 + k_8 = \beta_2 \\
    k_2 + k_4 + k_5=\beta_3\\
    k_3 = \beta_4 \\
    k_7 = \beta_5
  \end{cases} \\
  (*)' &\Leftrightarrow
  \begin{cases}
    k'_{2} + k'_{6} + k'_{8} + k'_{10} = \beta_1 \\
    k'_{1} + k'_{4} + k'_{8} + k'_{9} = \beta_2 \\
    k'_{3} + k'_{7} + k'_{8} + k'_{11} = \beta_3\\
    k'_{8} + k'_{9} + k'_{12} = \beta_4 \\
    k'_{4} + k'_{5} + k'_{8} = \beta_5
  \end{cases}
\end{align}
\end{example}

\end{document}